\def\year{2018}\relax
\documentclass[letterpaper]{article} 
\usepackage{aaai18}  
\usepackage{times}  
\usepackage{helvet}  
\usepackage{courier}  
\usepackage{url}  
\usepackage{graphicx}  
\usepackage{times}
\usepackage{graphicx}
\usepackage{subfigure}
\usepackage[skip=7pt,belowskip=0pt]{caption}
\usepackage{url}

\usepackage{pdfpages}

\usepackage{latexsym} 

\usepackage{amsmath,amssymb,amsthm} 
\usepackage[ruled,vlined,linesnumbered]{algorithm2e}
\usepackage{etoolbox}
\usepackage{multicol}
\newtheorem{theorem}{Theorem}
\newtheorem{proposition}{Proposition}
\newtheorem{definition}{Definition}

\def\@coverletter{no university}
\newcommand{\coverletter}[1]{
  \def\@coverletter{#1}
}

\title{Practical Scalability for Stackelberg Security Games}
\author{Arunesh Sinha$^\ast$, Aaron Schlenker$^\dagger$, Donnabell Dmello$^\dagger$, Milind Tambe$^\dagger$\\
$^\ast$University of Michigan, $^\dagger$University of Southern California\\
arunesh@umich.edu, \{aschlenk,ddmello,tambe\}@usc.edu}

\begin{document}

\maketitle

\begin{abstract}
  Stackelberg Security Games (SSGs) have been adopted widely for modeling adversarial interactions. With increasing size of the applications of SSGs, scalability of equilibrium computation is an important research problem. While prior research has made progress with regards to scalability, many real world problems cannot be solved satisfactorily yet as per current requirements; these include the deployed federal air marshals (FAMS) application and the threat screening (TSG) problem at airports. Further, these problem domains are inherently limited by NP hardness shown in prior literature. We initiate a principled study of approximations in zero-sum SSGs. Our contribution includes the following: (1) a unified model of SSGs called adversarial randomized allocation (ARA) games that allows studying most SSGs in one model, (2) hardness of approximation results for zero-sum ARA, as well as for the sub-problem of allocating federal air marshal (FAMS) and threat screening problem (TSG) at airports, (3) an approximation framework for zero-sum ARA with provable approximation guarantees and (4) experiments demonstrating the significant scalability of up to 1000x improvement in runtime with an acceptable 5\% solution quality loss.
\end{abstract}
\section{Introduction}
The Stackelberg Security Game (SSG) model (and variants) has been widely adopted in literature and in practice to model the defender-adversary interaction in various domains~\cite{Tambe2011,jiang2013defender,blum2014learning,guo2016coalitional,blocki2013,basilico2009leader,MunozdeCote}. Over time SSGs have been used to model increasingly large and complex real world problems, hence an important research direction in SSGs is the study of scalable Strong Stackelberg Equilibrium (SSE) computation algorithms, both theoretically and empirically. The scalability challenge has led to the development of a number of novel algorithmic techniques that compute the SSE of SSGs (see related work).  

However, scalability continues to remain a pertinent challenge across many SSG applications. 
There are real world problems that even the best known approaches fail to scale up to, such as threat screening games (TSGs) and the Federal Air Marshals (FAMS) domain. The TSG model is used to allocate screening resources to screenees at airports and solves the problem for every hour (24 times a day). Yet, recent state-of-the-art approach for airport threat screening~\cite{TSG2016} (TSG) scales only up to 110 flights per hour whereas 220 flights  can depart per hour from the Atlanta Airport~\cite{AirportCap}. The FAMS problem is to allocate federal air marshals to US based flights in order to protect against hijacking attacks. Again, the best optimal solver for FAMS in literature~\cite{jain2010security} solves problems up to 200 flights (FAMS is an already deployed application), whereas on average 3500 international flights depart from USA daily~\cite{USDOT}. Further, these problems are fundamentally limited by the hardness of computing the exact solution.

To overcome the computational hardness, and also provide practical scalability we investigate approximation techniques for zero-sum SSGs. Towards that end, our \emph{first contribution} in this paper is a \emph{unified} model of SSGs that we name \emph{adversarial randomized allocation} (ARA) games. ARA captures a large class of SSGs which we call linearizable SSGs (defined later) and it includes TSGs and FAMS. Further, the ARA model provides a unified characterization of implementability of marginal strategies that has been studied in separate papers in literature. 

Our \emph{second contribution} is a set of hardness of approximation results for the class of zero-sum ARA problems, and also for sub-classes such as TSGs. For ARA, we show that the ARA equilibrium computation problem and the defender best response problem in the given ARA game have the same hardness of approximation property and in the worst case ARA is not approximable. Further, we show that even the restricted set of ARA problems given by FAMS and TSGs are both hard to approximate to any logarithmic factor.

Our \emph{third contribution} is a general approximation framework for finding the SSE for zero-sum ARAs. As an application, we demonstrate the use of this framework for FAMS and TSGs. The approximation approach combines techniques from dependent sampling~\cite{tsai2010urban} with domain specific heuristics to yield simple to implement approximation algorithms. We provide theoretical approximation bounds for both FAMS and TSGs.

Finally, as our \emph{fourth contribution}, we demonstrate via experiments that we can solve FAMS problem up to 3500 flights and TSG problems up to 280 flights with runtime improvements up to 1000x. Moreover, the loss for FAMS problems is less than 5\% with increasing flights and for TSGs is less than 1.5\% across all cases. Hence, our results provide a practical and simple framework for approximating zero-sum SSGs (more generally ARAs) with provable approximation guarantees. Moreover, our approach enables solving the real world FAMS and airport screening problem satisfactorily for a US wide deployment. All full proofs are in the Appendix.



\section{Related Work}
Two major approaches to scale up in SSGs include incremental strategy generation (ISG) and use of marginals. ISG uses a master slave decomposition, with the slave providing a defender or attacker best response~\cite{jain2010security}. All these approaches are fundamentally limited by the computational complexity of finding an exact solution~\cite{KorzhykCP10,Xu16a} and thus in some cases ISG has included approximation of defender/attacker best responses~\cite{guo2016coalitional,gan2015security}. Use of marginals and directly sampling from marginals while faster suffers from the issue of non-implementable (invalid) marginal solutions~\cite{kiekintveld2009computing,tsai2010urban}. Fixing the non-implementability again runs into complexity barriers~\cite{TSG2016}. Combinations of marginals and ISG approaches~\cite{Bosansky2015} and gradient descent based approach~\cite{amingradient} have also been used. Our study stands in contrast to these approaches as we aim to approximate the SSE and not compute it exactly, providing a viable alternative to ISG and bypassing the non-implementability of marginals approach. Another line of work uses regret minimization and endgame solving techniques~\cite{Moravkeaam6960,brown2017safe} for approximately solving large scale sequential zero sum games. Our game does not have a sequential structure and the large action space ($10^{33}$ for TSG and $10^{14}$ for FAMS; see appendix for the calculation) precludes using a standard no-regret learning approach.


Our approximation approach is inspired by randomized rounding (RR)~\cite{raghavan1987randomized}. However, distinct from standard RR, our problem has equality constraints. Previous work on RR with equality constraints address \emph{only} equality constraints~\cite{gandhi2006dependent} or works on obtaining an integral solution given an approximate fractional solution within a polyhedron with integral vertices~\cite{ageev2004pipage,chekuri2009dependent}. However, our initial fractional solution may not lie in an integral polyhedron, and we have both equality and inequality constraints. Thus, we provide an approach that exploits the disjoint structure of equality constraints in TSGs in order to use previous work on comb sampling~\cite{tsai2010urban} and then alters the output~\cite{bansal2012solving} to handle both equality and inequality constraints.

While there exists other hardness of approximation results for other Stackelberg security game like models~\cite{guo2016coalitional} that rely on the graph characteristics of those domains, our hardness of approximation results are the first such results for the simple FAMS and the recent TSG security games problem. 

\section{Model and Notation}
We present a general abstract model of \emph{adversarial randomized allocation} (ARA). ARA captures all \emph{linearizable} SSGs, which is defined as those in which the probability $c_t$ of defending a target $t$ is linear in the defender mixed strategy. The ARA game model is a Stackelberg game model in which the defender moves first by committing to a randomized allocation and the adversary best responds. We start by presenting the defender's action space. There are $k$ defense assets that need to be allocated to $n$ objects to be defended. In this model, assets and objects are abstract entities and do not represent actual resources and targets in a security games. We will instantiate this abstract model with concrete examples of FAMS and TSG in the following sub-sections. 

\begin{figure}[t] 
  \centering
    \includegraphics[width=\linewidth]{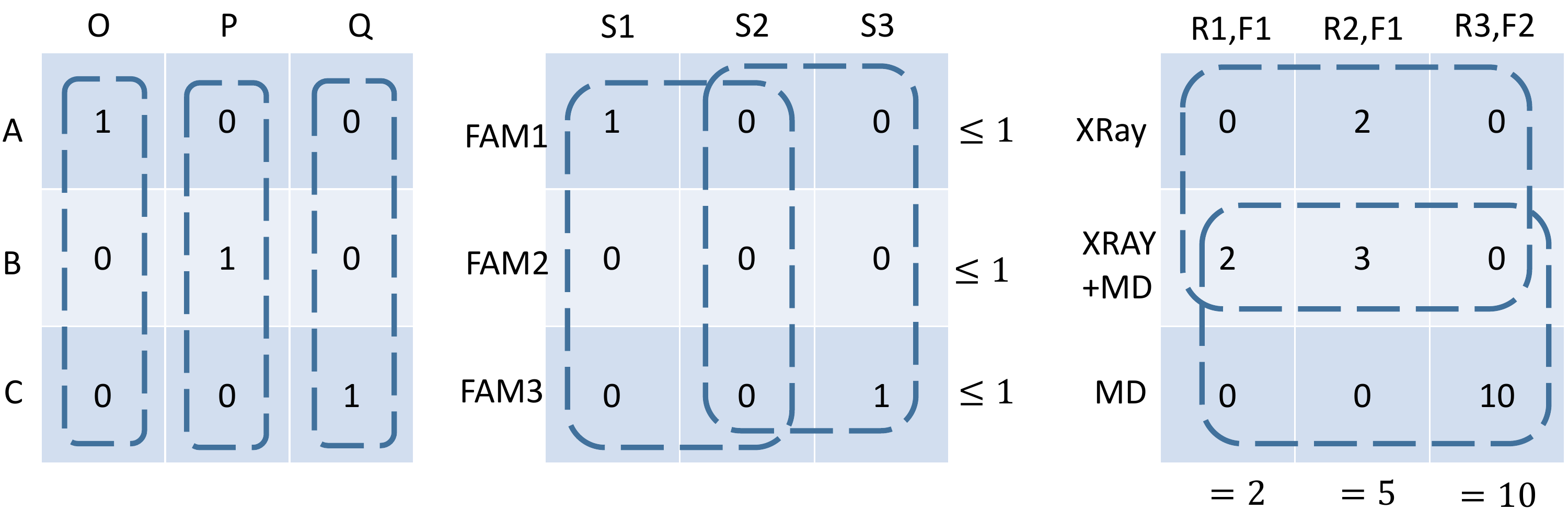}
    \caption{\small Three illustrations: (a) ARA with assets A,B,C and objects O,P,Q with 3 example assignment constraints (shown as dashed lines) with upper bound 1 on the columns. Shown also is an assignment that satisfies these constraints. (b) FAMS problem with 2 flights, 3 schedules and 3 FAMS. S1 and S2 share one flight and so do S2 and S3. The two assignment constraints (for the two flights) with upper bound 1 are represented by the two dashed lines.  Additional constraints are present on each row, shown on the right of the matrix. The attacker chooses a flight to attack, hence the dashed lines also show the index set $T$ of targets. A sample pure strategy fills the column entries. (c) TSG with the two assignment constraints (resource capacity) with upper bound 7 for XRay and 15 for Metal Detector (MD) represented by the two dashed lines. Additional equality constraints denoting the number of passengers in each passenger category (R,F) are present on each column, shown on the bottom of the matrix. A passenger category (column) is made from risk and flight. An adversary of type $R1$ can only choose the first column $R1,F1$ and $R2$ can choose from the other two columns. Thus, the index set $T$ for targets corresponds to columns. A sample pure strategy fills the column entries.}\label{Illus}
\end{figure}

\textbf{Defender's randomized allocation of resources}: The allocation can be represented as a $k \times n$ matrix with the $(i,j)^{th}$ entry $x_{i,j}$ denoting the allocation of asset $i$ to object $j$, and each $x_{i,j} \geq 0$. There is a set of \emph{assignment constraints} on the entries of the matrix. Each assignment constraint is characterized by a set $S \subseteq \{1, \ldots, k\}\times \{1, \ldots, n\}$ of indexes of the matrix and the constraint is given by $n_s \leq \sum_{(i,j) \in S}x_{i,j} \leq N_S$, where $n_s, N_S$ are non-negative integers. We will refer to each assignment constraint as $S$. Also for sake of brevity, we denote the vector of all the entries in the matrix as $\mathbf{x}$ and $\sum_{(i,j) \in S}x_{i,j}$ as $\mathbf{x}[S]$. 

Pure strategies of the defender are \emph{integral} allocations that respect the assignment constraints, i.e., \emph{integral} $\mathbf{x}$'s such that $n_S \leq \mathbf{x}[S] \leq N_S$ for all assignment constraints $S$. See Figure~\ref{Illus} for an illustrative example of the assignment constraints and a valid pure strategy.  Let the set of pure strategies be $P$ and we will refer to a single pure strategy as $\mathbf{P}$. On the other hand, the space of marginal strategies $MgS$ are those $\mathbf{x}$'s that satisfy the assignment constraints $n_S \leq \mathbf{x}[S] \leq N_S$ for all $S$; note that marginal strategies need not be integral. 

Mixed strategies are probability distributions over pure strategies, e.g., probabilities $a_1, \ldots, a_m$ ($\sum_m a_m = 1$) over pure strategies $\mathbf{P}_1, \ldots, \mathbf{P}_m$. An expected (marginal) representation of a mixed strategy is $\mathbf{x} = \sum_m a_m \mathbf{P}_m$. Thus, the space of mixed strategies is exactly the convex hull of $P$, denoted as $conv(P)$. 
Typically, the space of marginal strategies is larger than $conv(P)$, i.e., $conv(P) \subset MgS$, hence all marginal strategies are not implementable as mixed strategies. The conditions under which all marginal strategies are implementable (or not) has an easy interpretation in our model (see the implementability results later in this section).

\textbf{Adversary's action}: The presence of an adversary sets our model (and SSGs) apart from a randomized allocation problem~\cite{budish2013designing} and makes ARA a game problem. The attacker's action is to choose a target to attack. In our abstract formulation a target $t$ is given by a set $T \subset  \{1, \ldots, k\}\times \{1, \ldots, n\}$ of indexes of the allocation matrix. In order to capture linearizable SSGs, the probability of successfully defending an attack on target $t$ is $c_t = \sum_{i,j \in T} w_{i,j} x_{i,j}$ where $w_{i,j}$'s are given constants such that $w_{i,j} \leq 1/\max_{\mathbf{x} \in conv(P)} \sum_{i,j \in T} x_{i,j}$. The constraint on $w_{i,j}$ ensures that $c_t \leq 1$. We assume the total number of targets in polynomial in the size of the allocation matrix. Then, as is standard for SSGs, the defender utility for defender strategy $\mathbf{x}$ and attacker strategy $t$ is given by the expected value
$$ 
U_d(\mathbf{x}, t) = c_t U^t_s + (1- c_t) U^t_u
$$
where $U^t_s$ (resp. $U^t_u$) is the defender's utility when target $t$ is successfully (resp. unsuccessfully) defended. As we restrict ourselves to zero-sum games, the attacker's utility is negation of the above\footnote{We remark that modeling-wise the extension to general-sum case, non-linearity in probabilities or exponentially many targets is straightforward; here we restrict the model as it suffices for the domains we consider. See online appendix for the extension.}. 

The problem of Strong Stackelberg equilibrium computation can be stated as:
$\max_{\mathbf{x},z, a_i}   z$ subject to $z \leq U_d(\mathbf{x}, t) \; \forall t$ and 
$\mathbf{x} = \sum_{i:\mathbf{P}_i \in P} a_i \mathbf{P}_i $, where the last constraint represents $\mathbf{x} \in conv(P)$.
Note that the inputs to the SSE problem are the assignment constraints, and the number of pure strategies can be exponential in this input. 
Thus, even though the above optimization is a LP, its size can be exponential in the input to the SSE computation problem. However, using the marginal strategies $MgS$ instead of the mixed strategies $conv(P)$ results in a polynomial sized $ marginalLP$:
$$
\begin{array}{r l}
 \max_{\mathbf{x},z,c_t} &  z \\
 \mbox{subject to} & z \leq U(\mathbf{x}, t) \quad \forall t  \\
& n_s \leq \mathbf{x}[S] \leq N_S \;\; \forall S \mbox{ and } x_{i,j} \geq 0 \;\; \forall i,j
\end{array}
$$
But, as stated earlier $conv(P) \subset MgS$, and hence the solution to the optimization above may not be implementable as a valid mixed strategy. In our approximation approach we will solve the above $marginalLP$ as the first step obtaining marginal solution $\mathbf{x}^m$. 

\textbf{Bayesian Extension}\footnote{Typically player types denotes different utilities but as Harsanyi~\shortcite{harsanyi1967games} originally formulated, types capture any incomplete information including, as for our case, the lack of information about the action space of adversary. The game is still zero-sum.}: We also consider the following simple extension where we consider types of adversary $\theta \in \Theta$ and each adversary type $\theta$ attacks a set of targets $\mathcal{T}_{\theta}$ such that $\mathcal{T}_{\theta} \cap \mathcal{T}_{\theta'} = \phi$ for all $\theta, \theta' \in \Theta$. The adversary is of type $\theta$ with probability $p_\theta$ ($\sum_\theta p_\theta = 1$). Then, the exact SSE optimization can be written as: $\max_{\mathbf{x},z_{\theta}, a_i}   p_{\theta}z_{\theta}$ subject to $z_{\theta} \leq U_d(\mathbf{x}, t) \; \forall \theta \; \forall t \in \mathcal{T}_{\theta}$ and 
$\mathbf{x} = \sum_{i:\mathbf{P}_i \in P} a_i \mathbf{P}_i $. A corresponding $marginalLP$ can be defined in exactly the same way as for original ARA. Next, we show how the FAMS and TSG domain are instances of this abstract ARA model and Bayesian ARA respectively.

\textbf{Implementability}: Viewing the defender's action space as a randomized allocation provides an easy way to characterize non-implementability of mixed strategies across a wide range of SSGs, in contrast to prior work that have identified non-implementability for specific cases~\cite{KorzhykCP10,letchford2013solving,TSG2016} . The details of this interpretation can be found in the Appendix.

\subsection{FAMS}
We model zero-sum FAMS in the ARA model. The FAMS problem is to allocate federal air marshal (FAMS) to flights to and from US in order to prevent hijacking attacks. The allocation is constrained by the number of FAMS available and the fact that each FAMS must be scheduled on round trips that take them back to their home airport.
Thus, the main technical complication arises from the presence of schedules. A schedule is a subset of flights that has to be defended together, e.g., flight f1 and f2 should be defended together as they form a round trip for the air marshal. Air marshals are allocated to schedules, no flight can have more than one air marshal and some schedules cannot be defended by some air marshals. The adversary attacks a flight.

 Then, we capture the FAMS domain in the above model by mapping schedules in FAMS to objects (on columns) and air marshal in FAMS to assets (on rows). See Figure~\ref{Illus} for an illustrative example. The assignment constraints include the constraint for each resource $i$: $\sum_j x_{i,j} \leq 1$, which states that every resource can be assigned at most once. If an air marshal $i$ cannot be assigned to schedule $j$ then add the constraint $ x_{i,j} = 0$. A target $t$ in the abstract model maps to a flight $f$ in FAMS, and the set $T$ are all the indexes for all schedules that include this flight: $\{(i,j) ~|~$ flight f is in schedule $j\}$. The constraint that a flight cannot have more than one air marshal is captured by adding the \emph{target allocation constraint} $\mathbf{x}[T] \leq 1$. The probability of defending a target (flight) is $c_t = \mathbf{x}[T]$, hence the weights $w_{i,j}$'s in ARA are all ones. 
 
 \subsection{TSG}
We model TSGs using the Bayesian formulation of ARA. The TSG problem is how to allocate screening resources to screenees in order to screen optimally, which we elaborate in the context of airline passenger screening.
In TSGs, different TSG resources such as X-Rays, Metal Detector act in teams to inspect an airline passenger. The possible teams are given. Passengers are further grouped into passenger categories with a given $N_c$ number of passengers in each category $c$. The allocation is of resource teams to passenger categories. There are \emph{resource capacity constraints} for each resource usage (not on teams but on each resource). Further, all passengers need to be screened. Each resource team $i$ has an effectiveness $E_i < 1$ of catching the adversary. Observe that, unlike SSGs, the allocation in TSGs is not just binary $\{0,1\}$ but any positive integer within the constraints. The passenger category $c$ is a tuple of risk level and flight $(r, f)$; the adversary's action is to choose the flight $f$ but he is probabilistically assigned his risk level.
 
 Then, we capture the TSG domain in the above abstract model by mapping passenger categories in TSGs to objects (on columns) and resource teams in TSGs to assets (on rows). See Figure~\ref{Illus} for an illustrative example. The capacity constraint for each resource $r$ is captured by specifying the constraint $\mathbf{x}[S] \leq N_S$ which contains all indexes of teams that are formed using the given resource $r$: $S = \{(i,j) ~|~$ team $i$ is formed using resource $r\}$ with $N_S$ equal to the resource capacity bound for resource $r$. For every passenger category $j$, the constraint $\sum_i x_{i,j} = N_j$ enforces that all passengers are screened. A target $t$ in TSG is simply a passenger category $j$, thus, the set $T$ is $\{(i,j)~|~ j$ is given passenger category$\}$. The probability of detecting an adversary in category $j$ is given by $\sum_{(i,j) \in T} E_i x_{i,j}/N_j$, hence the weights $w_{i,j}$ are $E_i/N_j$; since $E_i < 1$ it is easy to check that $\sum_{(i,j) \in T} w_{i,j} x_{i,j} \leq 1$ for any $T$. The adversary type is the risk level $r$, and each type $r$ of adversary can choose a flight $f$, thus, choosing a  target which is the passenger category $(r,f)$. The probability of the adversary having a particular risk level is given.

\section{Computation Complexity}
In this section, we explore the \emph{hardness of approximation} for ARAs, FAMS and TSGs. In prior work on computation complexity of SSGs, researchers~\cite{Xu16a} have focused on hardness of exact computation providing general results relating the hardness of defender best response (DBR) problem (defined below) to the hardness of exact SSE computation. In contrast, we relate the hardness of approximation of the DBR problem to hardness of approximation of ARAs. We also prove that special cases of ARA such as FAMS and TSGs are also hard to approximate.

First, we formally state the equilibrium computation problem in adversarial randomized allocation: given the assets, objects and assignment constraints of an adversarial randomized allocation problem as input, output the SSE utility and a set of pure strategies $P_1, \ldots, P_m$ and probabilities $p_1, \ldots, p_m$ that represents the SSE mixed strategy.
We restrict $m$ to be polynomial in the input size. This is natural, since a polynomial time algorithm cannot produce an exponential size output. Also, as discussed in prior literature~\cite{Xu16a} the size of the support set of any mixed strategy is one more than the dimension of any pure strategy, which is the poly $kn+1$ in our case.

Next, as has been defined in prior literature~\cite{Xu16a}, we state the defender best response (DBR) problem which will help in understanding the results. The DBR problem can be interpreted as the defender's best response to a given mixed strategy of the adversary. The DBR problem also shows up naturally as the slave problem in column generation based approaches to SSGs. While it is easy to show the hardness of approximation of given DBR problems, the question of how its relates to the hardness of approximation of SSE computation is open. 
\begin{definition} The DBR problem is $\max_{\mathbf{x} \in P} \mathbf{d} \cdot \mathbf{x}$ where  $\mathbf{d}$ is a vector of positive constants. DBR is a combinatorial problem that takes the assignment constraints as inputs, and not $P$.
\end{definition}

Next, we state the standard definition of approximation
\begin{definition}
An algorithm for a maximization problem is $r$-approximate if it provides a feasible solution with value at least $OPT/r$, where $OPT$ is the exact maximum value.
\end{definition}
Note that lower $r$ means better approximation. Depending on the best $r$ possible, optimization problems are classified into various approximation complexity classes with increasing hardness of approximation in the following order PTAS, APX, log-APX, and poly-APX. We extensively use the well-known approximation preserving AP reduction between optimization problems for our results. 
We do not delve into the formal definition of complexity classes or AP reduction here due to lack of space and these concepts being standard~\cite{Ausiello}.
Our first result shows that the ARA's approximation complexity is same as that of the DBR problem and in the worst case cannot be approximated.

\begin{theorem} \label{ARAHard}
The following hardness of approximation results hold for ARA problems
\begin{itemize}
\item ARA problems cannot be approximated by any bounded factor in polynomial time, unless $P=NP$.
\item If the DBR problem for given ARA problem lies in some given approximation class (PTAS, APX, log-APX, poly-APX), then so does the ARA problem.
\end{itemize}
\end{theorem}

\begin{proof}[Proof Sketch] The first result works by constructing a ARA from a NP hard unweighted ($\mathbf{d}=1$) DBR problem such that the feasibility of the constructed ARA solves the DBR problem, thereby ruling out any approximation. Such unweighted DBR problems exist (for example for FAMS). The second part of the proof works by constructing an ARA problem with one target and showing that the solution yields an approximate value for a relaxed DBR with $\mathbf{x} \in conv(P)$. Moreover, this solution is an expectation over integral points (pure strategies), thus, at least one integral point in the support set output by ARA also provides an approximation for the corresponding combinatorial DBR.
\end{proof}

As the above complexity result is a worst case analysis, one may wonder whether the above result holds for sub-classes of ARA problems. We show that strong versions of inapproximatibility also holds for FAMS and TSGs.

\begin{theorem} \label{TSGInapp}
TSGs cannot be approximated to any logarithmic factor in poly time, unless P=NP.
\end{theorem}

\begin{theorem} \label{FAMSHard}
FAMS problems cannot be approximated to any logarithmic factor in poly time, unless P=NP.
\end{theorem}

Both the proofs above use AP reduction from max independent set. The proof for TSG follows from an observation that a special case of the TSG problem is the independent set problem itself, thus, any approximation for TSGs can provide an equivalent approximation for the max independent set problem. The proof for FAMS is much more involved. It involves constructing a poly number of FAMS instance (with varying number of resources) given any independent set problem. The FAMS instances are such that pure strategies corresponds to independent sets and the multiple instances are constructed such that the optimal exact solution for one of these instances provides a solution for the max independent set problem. Then, it is shown that any approximation $r$ for the FAMS problem yields at least one pure strategy (over all the FAMS instances) that corresponds to a better than $r$ approximation for the max independent set problem, thereby completing the AP reduction. See the supplementary material for details.

\section{Approximation approach}
Our approach to approximation first solves the $marginalLP$, which is quite fast in practice (see experiments) and provides an upper bound to the true value of the game. Then, we sample from the marginal solution, but unlike previous work~\cite{tsai2010urban}, we alter the sampled value to ensure that the final pure strategy output is valid. 
We describe an abstract sampling and alteration approach for ARA in this part, which we instantiate for FAMS and TSGs in the subsequent sub-sections. Recall that a constraint is given by an index set $S$ and the constraint is an equality if $n_S = N_S$. For our abstract approach we restrict our attention to ARAs with \emph{partitioned equality assignment constraints}, which means the index set $S$ for all equality constraints partitions the index set $\{1, \ldots, k\}\times \{1, \ldots, n\}$ of the allocation matrix. Further, for inequality constraints we assume $n_S = 0$. Call these problems as PE0-ARA; this class still includes FAMS and TSGs. For FAMS, which does not have equality constraints, we use dummy schedules $s_i$ to get partitioned equalities $\sum_j x_{i,j} + s_i = 1$; $s_i = 1$ denotes that resource $i$ is unallocated. Our abstract approximation approach for PE0-ARA is presented in Algorithm~\ref{abstractalgorithm}.

\begin{algorithm}[t]
\caption{Abstract Approximation}\label{abstractalgorithm}
\DontPrintSemicolon
\ForAll{$S \in $ EqualityConstraints}{
$\mathbf{x} \leftarrow CombSample(\mathbf{x}^m,S)$
} 
$\mathbf{x} \leftarrow FixViolatedInequalityConstraints(\mathbf{x})$ \;
$\mathbf{x} \leftarrow FixEqualityConstraints(\mathbf{x})$ \;
\end{algorithm}

The Algorithm takes as input the marginal solution $\mathbf{x}^m$ from $marginalLP$ and produces a pure strategy. The first for loop (line 1-2) performs comb sampling for each equality constraint $S$ to produce integral values for the variables involved in $S$. Comb sampling was introduced in an earlier paper~\cite{tsai2010urban}; it provides the guarantee that $x_{i,j}^m$ is rounded up or down for all $(i,j) \in S$, the sample $x_{i,j}$ has expected value $E(x_{i,j}) = x^m_{i,j}$ for all $(i,j) \in S$ and equality S is still satisfied after the sampling. See Figure~\ref{IllusAlgo} for an example. Briefly, comb sampling works by creating $Z$ buckets of length one each, where $Z = \sum_{(i,j) \in S} \{x^m_{i,j}\}$, where $\{.\}$ denotes fractional part. Each of the $\{x^m_{i,j}\}$ length fraction is packed into the bucket (in any order and some of the $\{x^m_{i,j}\}$ fraction may have to be split into two buckets), then a number between $[0,1]$ is sampled randomly, say $z$, and for each bucket a mark is put at length $z$. Finally, the $(i,j)$ whose $\{x^m_{i,j}\}$ fraction lies on the marker $z$ for each bucket is chosen to be rounded up, and all other $x^m_{i,j}$ are rounded down. 

Observe that in expectation the output of comb sampling matches the marginal solution, thus, providing the same expected utility as the marginal solution. Recall that this expected utility is an upper bound on the optimal utility. However, the samples from comb sampling may not be valid pure strategies.
Thus, in case the output of comb sampling is not already valid, the two abstract methods in line 3 and 4 modify the sample strategy by first decreasing the integral values to satisfy the violated inequalities and then increasing the integral values to satisfy the equalities. Such modification of the sampled strategy to obtain a valid strategy is guided by the principle that the change in defender utility between the sampled and the resultant valid strategy should be small, which ensures that change in expected utility from the marginal solution due to the modification is small. As the expected utility of the marginal solution is an upper bound on the optimal expected utility this marginal expected utility guided modification leads the output expected utility to be close to the optimal utility. 

These two methods on line 3 and 4 are instantiated with domain specific heuristics that implement the principle of marginal expected utility guided modification. Below, we show the instantiation for the TSG and FAMS domains. A sample execution of the above approach for TSGs is shown in Figure~\ref{IllusAlgo}. 

\subsection{TSG}
\begin{figure}[t] 
  \centering
    \includegraphics[width=\linewidth]{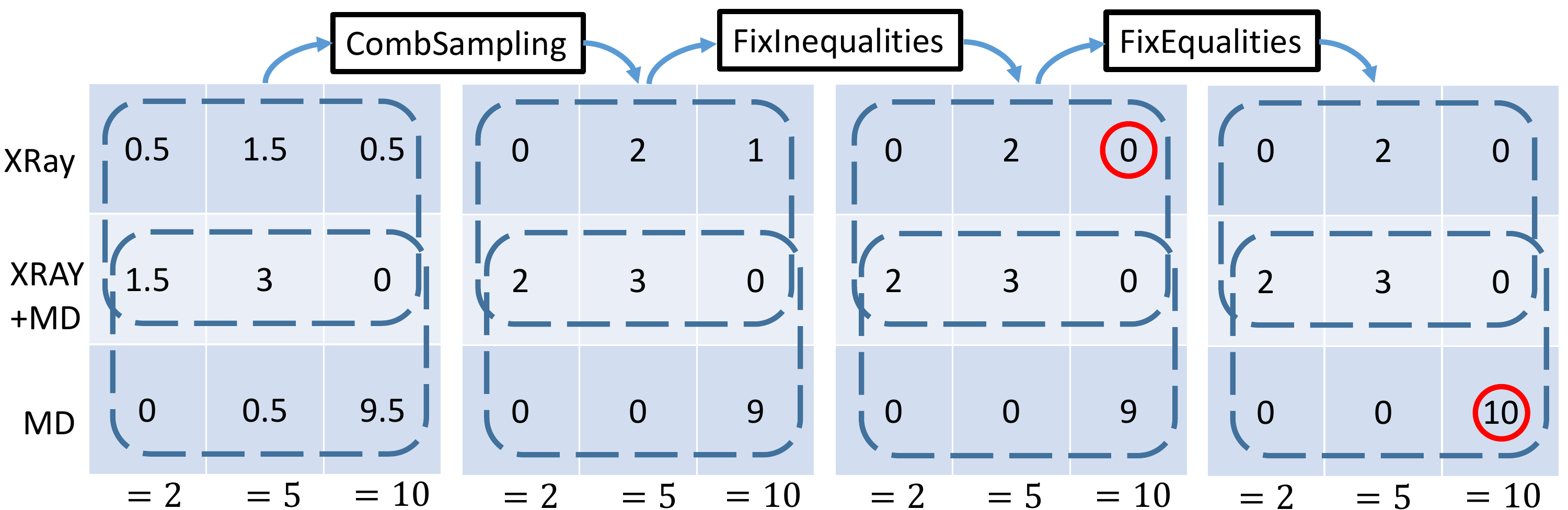}
    \caption{\small (Left to right) Sample execution for TSG: This uses the same TSG example as in Figure~\ref{Illus}. The marginal solution is the leftmost matrix which after CombSampling on each column becomes integral, e.g., 0.5 in the left column is rounded down to 0 and 1.5 rounded up to 2. Note that the CombSampling output satisfies all equalities, but exceeds the resource capacity 7 for X-ray. Next, allocation values are lowered (shown as red circle) to satisfy the X-Ray capacity but the equality constraint on third column gets violated. Next, allocation values are increased (again red circle) to fix the equality and that produces a valid pure strategy.}\label{IllusAlgo}
\end{figure}

The heuristics for TSG are guided by three \emph{observations}: (1) more effective resources are more constrained in their usage, (2) changing allocation for passenger categories with higher number of passengers changes the probability of detection of adversary by a smaller amount than changing allocation for category with fewer passengers and (3) higher risk passenger categories typically have lower number of passengers.

Recall that for TSGs the inequalities are resource capacity constraints.
Thus, for fixing violated inequalities we need to decrease allocation which decreases utility; we wish to keep the utility decrease small as that ensures that the expected utility does not move much further away from the upper bound marginal expected utility. Our approach for such decrease in allocation has the following steps: (a) prioritize fixing inequality of most violated resources first and (b) for each such inequality we attempt to lower allocation for passenger category with higher number of passengers. In light of the observations for TSG above this approach aims to keep the change in expected utility small. Specifically, observation 1 makes it likely that constraints for more effective resources are fixed in step a above. Observation 3 suggests that the changes in step b happens for lower risk passengers. Thus, step a aims to keep the allocation of effective resources for high risk passengers unchanged. This keeps the utility change small as changing allocation for high risk passengers can change utility by a large amount.
Next, by observation 2, step b aims to minimize the change in probability of detecting the adversary by a low amount so that expected utility change in small.
 For example in Figure~\ref{IllusAlgo}, the inequality fix reduces the allocation for the third passenger category (column) which also has the highest number of passengers (15).

Next, the equalities in TSGs are the constraints for every passenger category. For fixing equalities we need to increase allocation which increases utility; we wish to keep this utility increase high as it brings the expected utility closer to the upper bound marginal expected utility. Here we aim to do so by (a) prioritizing increase of allocation for categories with fewer people and (b) increasing allocation of those resources that have least slack in their resource capacity constraint (low slack could mean higher effectiveness). By Observation 1 low slack means that resource could be more effective and by Observation 2 fewer people means higher risk passengers. This ensures that higher risk passengers are screened more using more effective resources thereby raising the utility maximally. For example in Figure~\ref{IllusAlgo}, the equality for the third column is fixed by using the only available resource MD.

Recall that TSGs differs from FAMS in that the allocation for TSGs can be non-binary. This offers an advantage for TSGs with respect to approximation, as small fractional changes in allocation do not change the overall allocation by much (0.5 to 1 is a 50\% change in FAMS, but 4.5 to 5 is less than 10\% change for TSGs). Thus, we assume here that the changes due to Algorithm~\ref{abstractalgorithm} do not reduce the probability of detecting an adversary in any passenger category (from the marginal solution) by more than $1/c$ factor, where $c > 1$ is a constant. This restriction is realistic as it is very unlikely that any passenger category will have few passengers and we only aim to change the allocation for passenger categories with a higher number of passengers. Hence we prove the following
\begin{theorem} \label{TSGApp}
Assume that Algorithm~\ref{abstractalgorithm} always successfully outputs a pure strategy and the change in allocation from the marginal strategy does not change the probability of detecting an adversary by more than $1/c$ factor.
Then, the approximation approach above with the heuristic provides a $c$-approximation for TSGs.
\end{theorem}

As a remark, the above result does not violate the inapproximatability of TSGs since the above result holds for a restricted set of TSG problems. Also, the approximation for TSGs may sometimes fail to yield a valid pure strategy as satisfying the equalities may become impossible after using certain sequences of decreasing allocation. In our experiments we observe that the failure of obtaining a pure strategy after Algorithm~\ref{abstractalgorithm} is rare and also easily handled by repeating the approximation approach with a new sample (approximation runs in milli-secs).

\subsection{FAMS}
Recall that for FAMS the inequalities are the target allocation constraints: $\mathbf{x}[T] \leq 1$ and fixing violations for these involves decreasing allocation. Our heuristic is simple: the variables $x_{i,j}$ are set to zero (i.e., decreased) starting from those schedules $j$ that contain the most number of targets for which target allocation constraint is violated and do not contain any target for which the target allocation constraint is satisfied. We can guarantee to find a decrease in allocation that satisfies the constraint for $T$ without changing the allocation for targets that already satisfy constraints in the cases when the target $T$ with violated constraint (1) belongs to a schedule $j$ that exclusively contains that target $T$ ($x_{i,j}$ can be decreased without affecting any other constraint) or (2) $T$ belongs to only one schedule (all other targets in this schedule will violate their constraints). This approach ensures that we only work to fix the violated constraints and cause a minimal change in utility by leaving the satisfied constraints undisturbed. However, 
if in fixing a violated target allocation constraint for $T$ it becomes necessary to reduce allocation for another already satisfied target constraint, then sample uniformly from the $\geq 2$ schedules that $T$ belongs to in order to choose the $x_{i,j}$ allocation to reduce. Do this till all inequality constraints are satisfied. 

Then, we do nothing to fix equality constraints since we have only decreased $x_{i,j}$ and if any equality $\sum_j x_{i,j} + s_i = 1$ is not satisfied we can always set the dummy $s_i$ to be one. Also, observe that since we only always decrease allocations, we always find a pure strategy for any sample from Algorithm 1 (unlike TSGs). We prove the following:

\begin{theorem} \label{FAMSApp}
Let $C_t$ be the number of targets that share a schedule with any target $t$, and $C = \max_t C_t$.
The approximation approach above with the heuristic provides a $2^Ck$-approximation for FAMS.
\end{theorem}

Next, we show experimental results that reveal the average case loss of our approximation approach (as opposed to the worst case approximation guarantees in this section).

\section{Experimental Results}
\begin{figure}[t]
    \centering
    \subfigure[Runtime (log-scale time)]{
        \includegraphics[width=0.47\columnwidth]{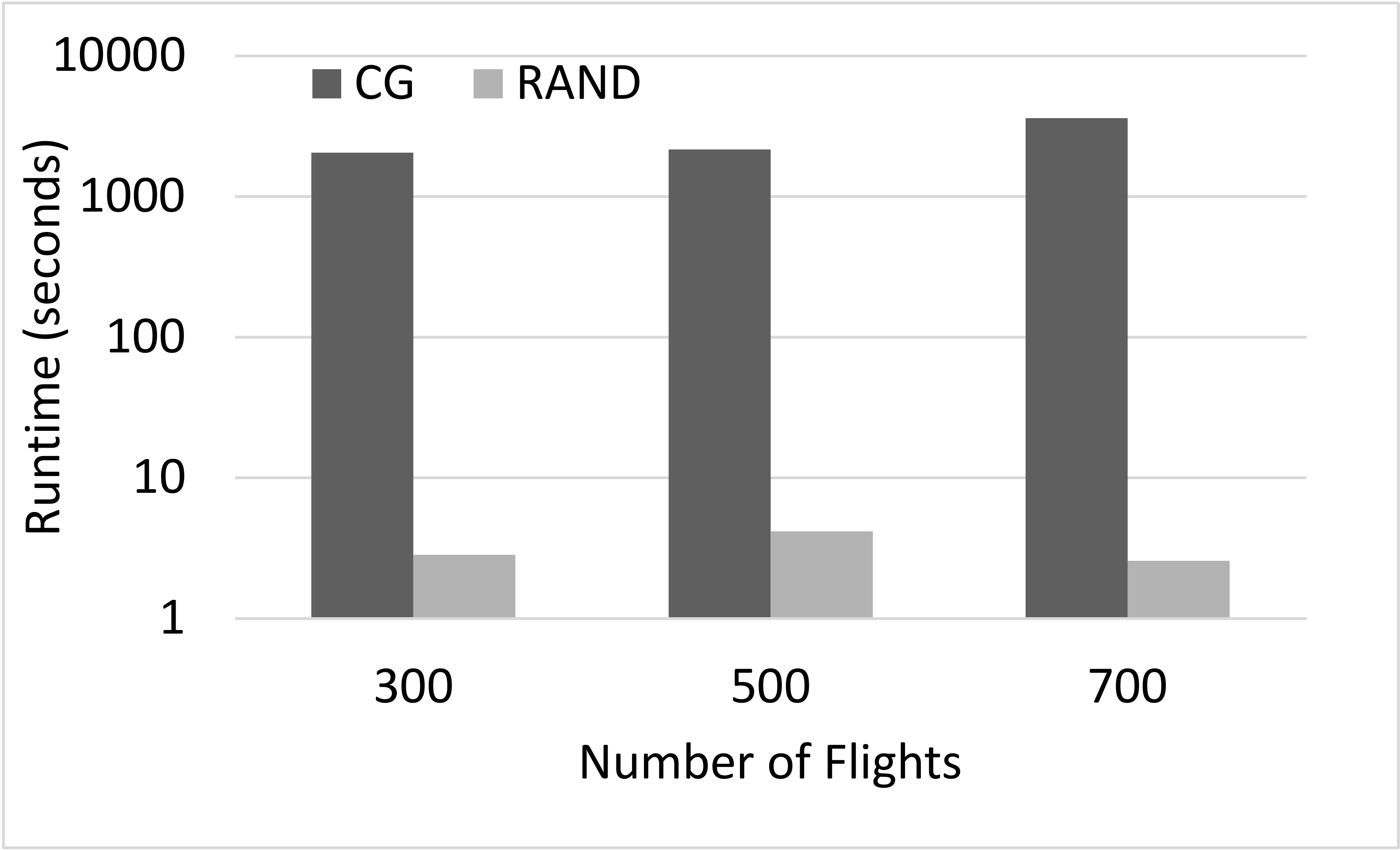}
        \label{fig:FAMSRunTime}
    }
    \subfigure[Solution quality]{
        \includegraphics[width=0.47\columnwidth]{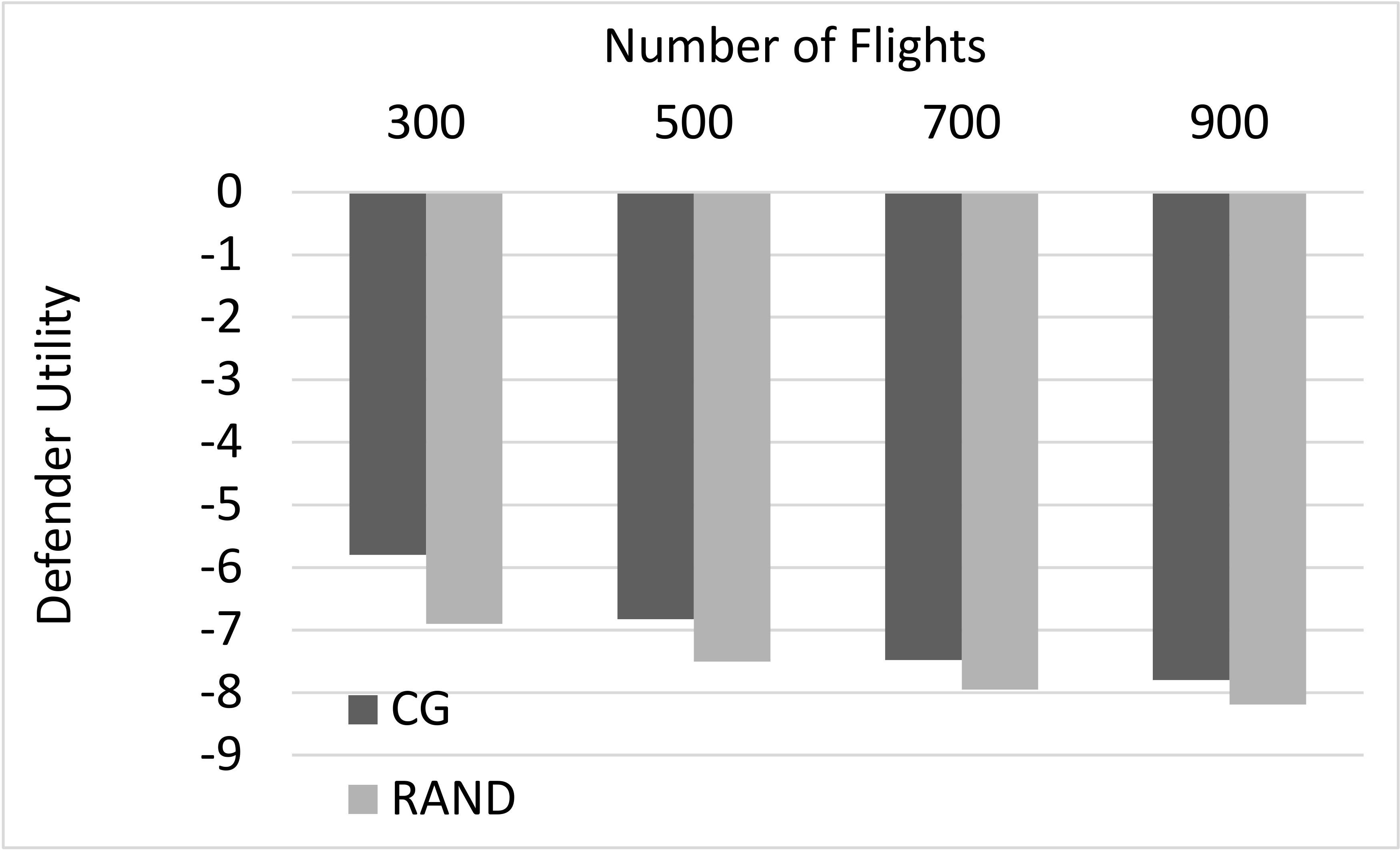}
        \label{fig:FAMSSol} 
    }
    \caption{ASPEN and RAND Comparison} \label{AMS}
\end{figure}

\begin{figure}[t]
    \centering 
    \subfigure[Runtime (log-scale time)]{
        \includegraphics[width=0.47\columnwidth]{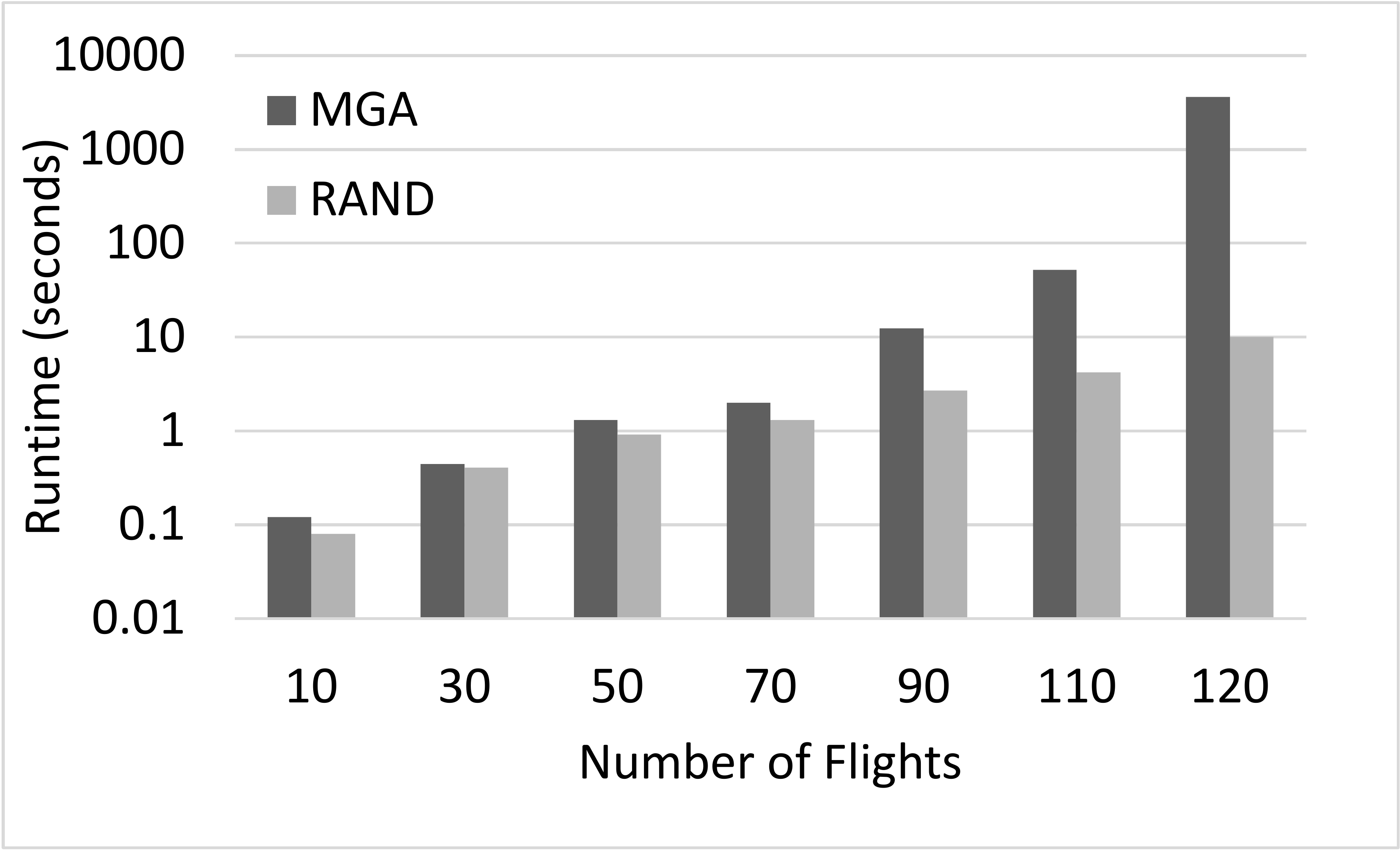}
        \label{fig:MGARunTime}
    }
    \subfigure[Solution quality]{
        \includegraphics[width=0.47\columnwidth]{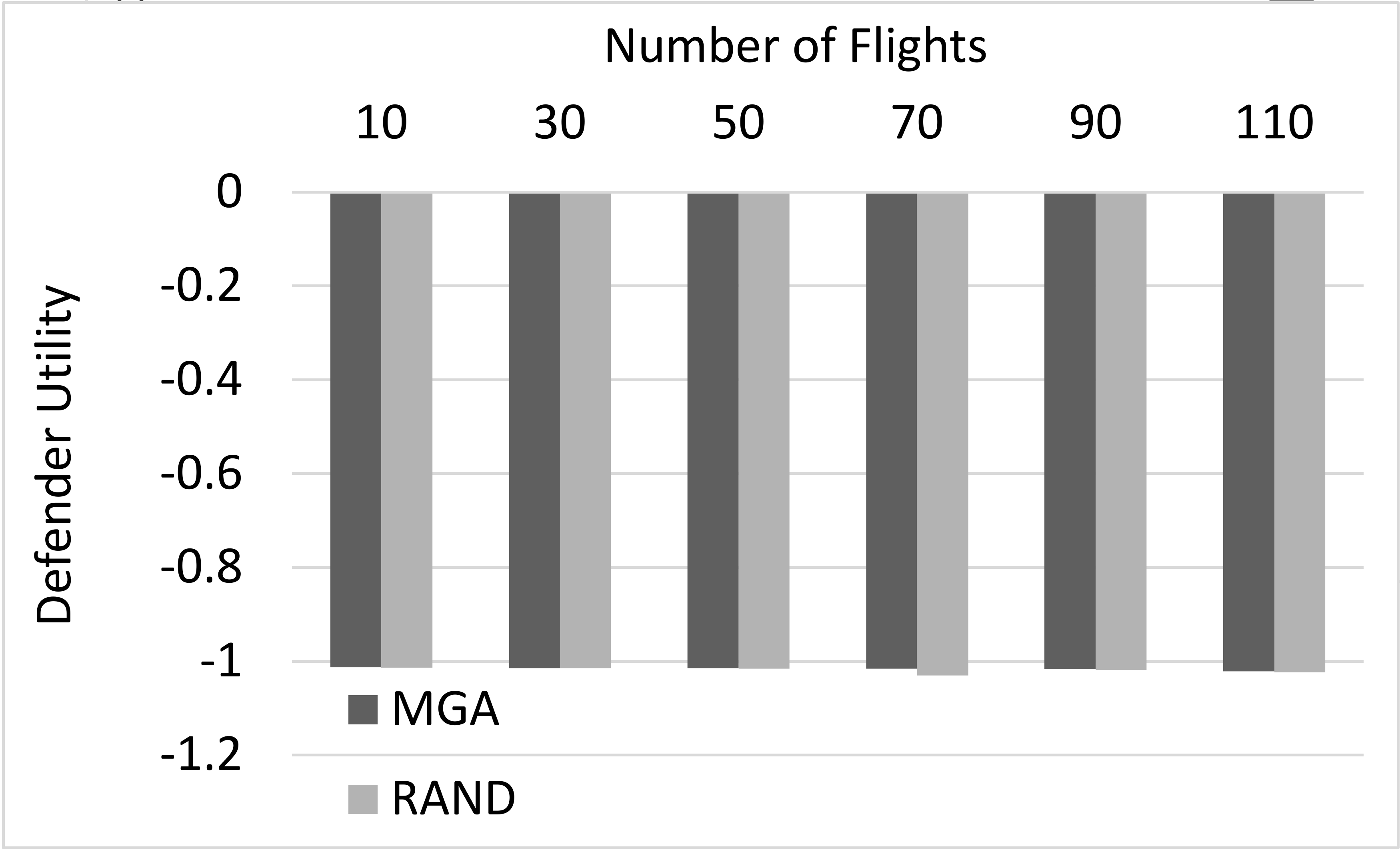}
        \label{fig:MGASol}
    }
    \caption{MGA and RAND Comparison} \label{MGA}
\end{figure}

Our set of experiments provide a comprehensive analysis of our randomized approach, which we name RAND. We compare RAND to the best know solver for TSGs called MGA; MGA~\cite{TSG2016} has been previously shown to outperform column generation based approaches by a large margin. For the FAMS problem the best known solver in literature for the general sum case is called the ASPEN~\cite{jain2010security}, which is a column generation based branch and price approach. Through private communication with the company (Avata Intelligence) managing the FAMS software, we know the FAMS problem is solved as a zero-sum problem for scalability and takes hours to complete. For our zero-sum case we implemented a plain column generation (CG) solver for FAMS, since branch and price is an overkill for the zero sum case. All results in this section are averaged over 30 randomly generated game instances. All game instances fix $U^t_s$ to $-1$ and randomly select $U^t_u$ from integers between $-2$ and $-10$. The utility for RAND is computed by generating 1000 pure strategies and taking their average as an estimate of the defender mixed strategy. All experiments
were run using a system with Xeon 2.6 GHz processor and 4GB RAM. 

For FAMS, we vary the number of flights, keeping the number of resources fixed at 10 and number of schedules fixed at 1000 and 5 targets/schedule. The runtime \emph{in log scale} is shown in Figure~\ref{AMS}. CG hits the 60 min cut-off for 700 flights and the run time for RAND is much lower at only a few seconds. 
Next, we report the solution quality for RAND by comparing with the solution using  CG. It can be seen that the solution gets better with increasing flights starting from 19\% loss at 300 flights to 5\% loss at 900 flights. Thus, the numbers show that we obtain large speed-ups up to factor of 1000x and are still able to extract 95\% utility for higher number of flights when the CG approach starts taking huge amount of time.

For TSGs, we used six passenger risk levels, eight screening resource types and 20
screening team types. We vary the number of fights and we also randomly sample the team structure (how teams are formed from resources) for each of the 30 runs. The results in Figure~\ref{MGA} show runtime (in log scale) and defender utility values varying with number of flights (on x-axis). As can be seen, MGA only scales up to 110 flights before hitting the cut-off of 60 mins, while RAND takes only 10 sec for 110 flights. Also, the solution quality loss for RAND has a maximum averaged loss of 1.49\%. Thus, we obtain at-least 360X speed-ups with very minor loss.

Finally, we test the scalability of RAND for FAMS and TSG, shown in Figure~\ref{ScaleRAND}. As can be seen, the runtime for RAND is low even with the highest number of flights we tested: 280 for TSG and 3500 for FAMS. The maximum runtime for FAMS was under 5 seconds. With TSGs the maximum runtime was under 25 secs. Overall, these results demonstrate the scalability of our approach with minor solution quality loss on average. Importantly, the scalability results shows that RAND can solve the FAMS and TSG problem for sizes that are required for a US wide deployment, and the results also provide evidence that further scalability with increasing number of flights is also possible. This also provides evidence that approximations is a viable approach to solving zero-sum SSGs at a large scale.

\begin{figure}[t]
    \centering
    
    \subfigure[Runtime FAMS]{
        \includegraphics[width=0.47\columnwidth]{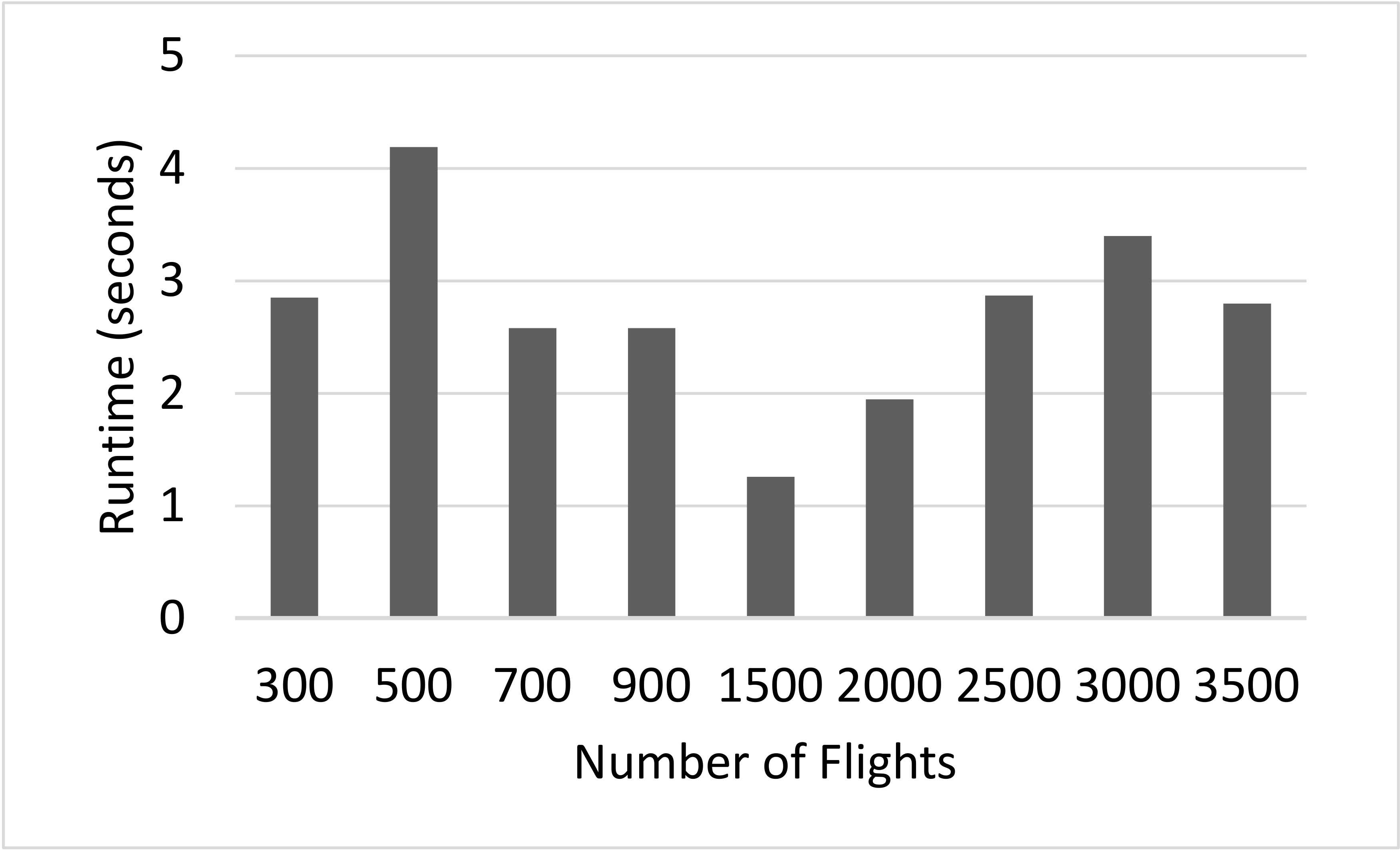}
        \label{fig:FAMSRRRuntime}
    }
    \subfigure[Runtime TSG]{
        \includegraphics[width=0.47\columnwidth]{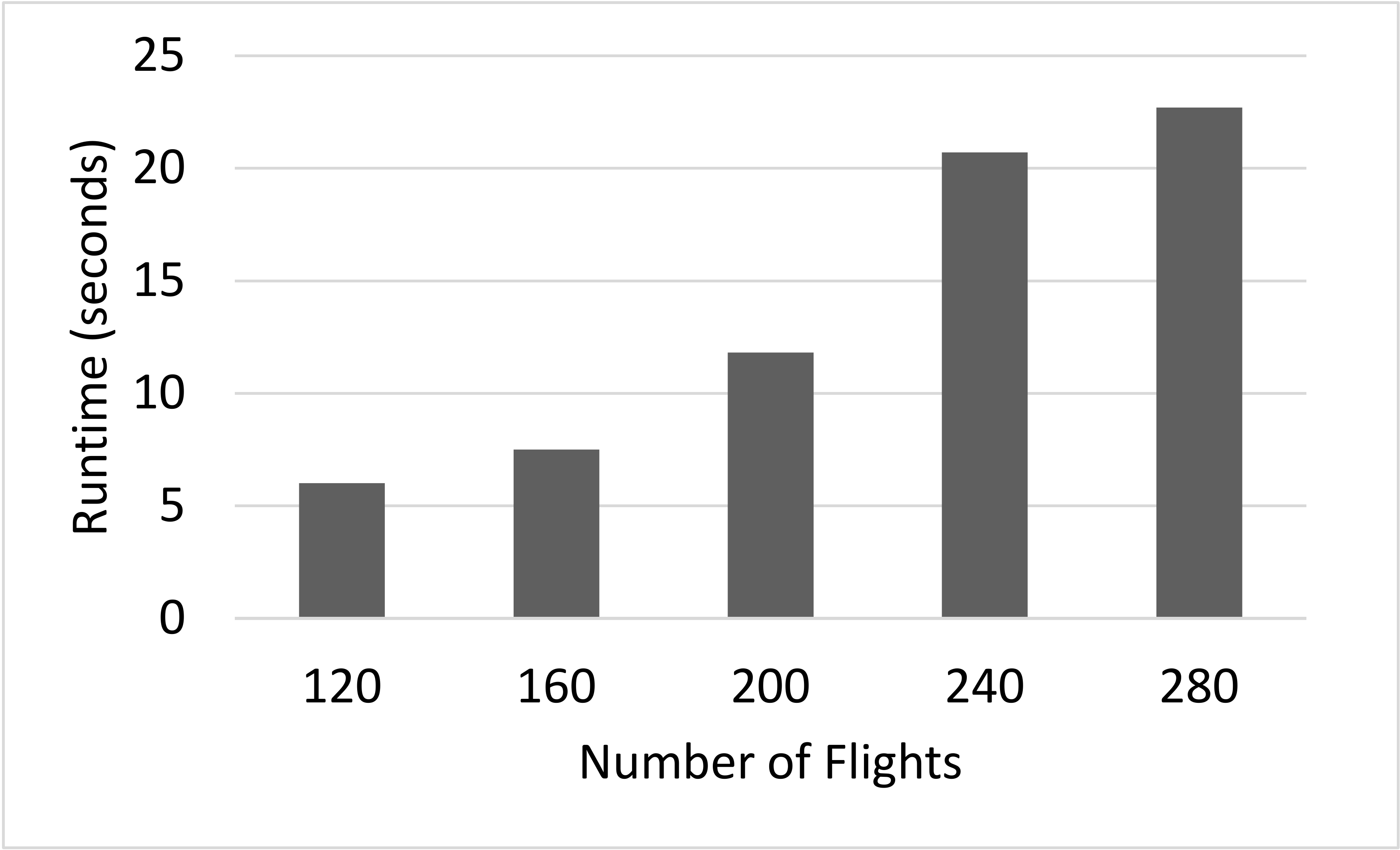}
        \label{fig:TSGRRRunTime}
    }
    \caption{Scalability of RAND} \label{ScaleRAND}
\end{figure}

\section{Conclusion}
We studied approximations in SSGs both theoretically and practically. In fact, while our paper title talks about half a loaf (optimal solution), we practically obtain up to 97\% of the loaf (optimal solution). We obtain 1000x speedups with losses within 5\%. Thus, approximations open the door to solving very large scale zero-sum SSGs and we believe that we have laid a fertile ground for future research in approximating large scale zero-sum security games. Our approach also provided an avenue to solve the real world FAMS and airport screening problem in order to increase the scope of applicability of these already deployed application (in case of FAMS) or applications under test (airport screening).

\clearpage
\bibliographystyle{aaai}
\bibliography{aaai18}

\begin{thebibliography}{}

\bibitem[\protect\citeauthoryear{Ageev and Sviridenko}{2004}]{ageev2004pipage}
Ageev, A.~A., and Sviridenko, M.~I.
\newblock 2004.
\newblock Pipage rounding: A new method of constructing algorithms with proven
  performance guarantee.
\newblock {\em Journal of Combinatorial Optimization} 8(3):307--328.

\bibitem[\protect\citeauthoryear{Amin, Singh, and Wellman}{2016}]{amingradient}
Amin, K.; Singh, S.; and Wellman, M.
\newblock 2016.
\newblock Gradient methods for stackelberg security games.
\newblock In {\em Conference in Uncertainty in Artificial Intelligence}.

\bibitem[\protect\citeauthoryear{Ausiello \bgroup et al\mbox.\egroup
  }{1999}]{Ausiello}
Ausiello, G.; Protasi, M.; Marchetti-Spaccamela, A.; Gambosi, G.; Crescenzi,
  P.; and Kann, V.
\newblock 1999.
\newblock {\em Complexity and Approximation: Combinatorial Optimization
  Problems and Their Approximability Properties}.

\bibitem[\protect\citeauthoryear{Bansal \bgroup et al\mbox.\egroup
  }{2012}]{bansal2012solving}
Bansal, N.; Korula, N.; Nagarajan, V.; and Srinivasan, A.
\newblock 2012.
\newblock Solving packing integer programs via randomized rounding with
  alterations.
\newblock {\em Theory of Computing} 8(1):533--565.

\bibitem[\protect\citeauthoryear{Basilico, Gatti, and
  Amigoni}{2009}]{basilico2009leader}
Basilico, N.; Gatti, N.; and Amigoni, F.
\newblock 2009.
\newblock Leader-follower strategies for robotic patrolling in environments
  with arbitrary topologies.
\newblock In {\em AAMAS},  57--64.

\bibitem[\protect\citeauthoryear{Blocki \bgroup et al\mbox.\egroup
  }{2013}]{blocki2013}
Blocki, J.; Christin, N.; Datta, A.; Procaccia, A.; and Sinha, A.
\newblock 2013.
\newblock Audit games.
\newblock In {\em IJCAI}.

\bibitem[\protect\citeauthoryear{Blum, Haghtalab, and
  Procaccia}{2014}]{blum2014learning}
Blum, A.; Haghtalab, N.; and Procaccia, A.~D.
\newblock 2014.
\newblock Learning optimal commitment to overcome insecurity.
\newblock In {\em NIPS},  1826--1834.

\bibitem[\protect\citeauthoryear{Bo\v{s}ansk\'{y} \bgroup et al\mbox.\egroup
  }{2015}]{Bosansky2015}
Bo\v{s}ansk\'{y}, B.; Jiang, A.~X.; Tambe, M.; and Kiekintveld, C.
\newblock 2015.
\newblock Combining compact representation and incremental generation in large
  games with sequential strategies.
\newblock In {\em AAAI}.

\bibitem[\protect\citeauthoryear{Brown and Sandholm}{2017}]{brown2017safe}
Brown, N., and Sandholm, T.
\newblock 2017.
\newblock Safe and nested subgame solving for imperfect-information games.
\newblock {\em arXiv preprint arXiv:1705.02955}.

\bibitem[\protect\citeauthoryear{Brown \bgroup et al\mbox.\egroup
  }{2016}]{TSG2016}
Brown, M.; Sinha, A.; Schlenker, A.; and Tambe, M.
\newblock 2016.
\newblock One size does not fit all: A game-theoretic approach for dynamically
  and effectively screening for threats.
\newblock In {\em AAAI}.

\bibitem[\protect\citeauthoryear{Budish \bgroup et al\mbox.\egroup
  }{2013}]{budish2013designing}
Budish, E.; Che, Y.-K.; Kojima, F.; and Milgrom, P.
\newblock 2013.
\newblock Designing random allocation mechanisms: Theory and applications.
\newblock {\em The American Economic Review} 103(2):585--623.

\bibitem[\protect\citeauthoryear{Chekuri, Vondr{\'a}k, and
  Zenklusen}{2009}]{chekuri2009dependent}
Chekuri, C.; Vondr{\'a}k, J.; and Zenklusen, R.
\newblock 2009.
\newblock Dependent randomized rounding for matroid polytopes and applications.
\newblock {\em arXiv preprint arXiv:0909.4348}.

\bibitem[\protect\citeauthoryear{FAA}{2014}]{AirportCap}
FAA.
\newblock 2014.
\newblock Airport capacity profiles.
\newblock
  {\small\url{https://www.faa.gov/airports/planning_capacity/profiles/media/Airport-Capacity-Profiles-2014.pdf}}.
\newblock Accessed: 2017-02-15.

\bibitem[\protect\citeauthoryear{Gan, An, and
  Vorobeychik}{2015}]{gan2015security}
Gan, J.; An, B.; and Vorobeychik, Y.
\newblock 2015.
\newblock Security games with protection externalities.
\newblock In {\em AAAI}.

\bibitem[\protect\citeauthoryear{Gandhi \bgroup et al\mbox.\egroup
  }{2006}]{gandhi2006dependent}
Gandhi, R.; Khuller, S.; Parthasarathy, S.; and Srinivasan, A.
\newblock 2006.
\newblock Dependent rounding and its applications to approximation algorithms.
\newblock {\em Journal of the ACM (JACM)} 53(3):324--360.

\bibitem[\protect\citeauthoryear{Guo \bgroup et al\mbox.\egroup
  }{2016}]{guo2016coalitional}
Guo, Q.; An, B.; Vorobeychik, Y.; Tran-Thanh, L.; Gan, J.; and Miao, C.
\newblock 2016.
\newblock Coalitional security games.
\newblock In {\em AAMAS}.

\bibitem[\protect\citeauthoryear{Harsanyi}{1967}]{harsanyi1967games}
Harsanyi, J.
\newblock 1967.
\newblock Games with incomplete information played by" bayesian" players, i-iii
  part i. the basic model.
\newblock {\em Management Science} 14(3):159--182.

\bibitem[\protect\citeauthoryear{Jain \bgroup et al\mbox.\egroup
  }{2010}]{jain2010security}
Jain, M.; Karde{\c{s}}, E.; Kiekintveld, C.; Tambe, M.; and Ord{\'o}{\~n}ez, F.
\newblock 2010.
\newblock Security games with arbitrary schedules: a branch and price approach.
\newblock In {\em AAAI},  792--797.

\bibitem[\protect\citeauthoryear{Jiang \bgroup et al\mbox.\egroup
  }{2013}]{jiang2013defender}
Jiang, A.~X.; Procaccia, A.~D.; Qian, Y.; Shah, N.; and Tambe, M.
\newblock 2013.
\newblock Defender (mis) coordination in security games.
\newblock AAAI.

\bibitem[\protect\citeauthoryear{Kiekintveld \bgroup et al\mbox.\egroup
  }{2009}]{kiekintveld2009computing}
Kiekintveld, C.; Jain, M.; Tsai, J.; Pita, J.; Ord{\'o}{\~n}ez, F.; and Tambe,
  M.
\newblock 2009.
\newblock Computing optimal randomized resource allocations for massive
  security games.
\newblock In {\em AAMAS},  689--696.

\bibitem[\protect\citeauthoryear{Korzhyk, Conitzer, and
  Parr}{2010}]{KorzhykCP10}
Korzhyk, D.; Conitzer, V.; and Parr, R.
\newblock 2010.
\newblock Complexity of computing optimal {S}tackelberg strategies in security
  resource allocation games.
\newblock In {\em AAAI}.

\bibitem[\protect\citeauthoryear{Letchford and
  Conitzer}{2013}]{letchford2013solving}
Letchford, J., and Conitzer, V.
\newblock 2013.
\newblock Solving security games on graphs via marginal probabilities.
\newblock In {\em AAAI}.

\bibitem[\protect\citeauthoryear{Morav{\v c}{\'\i}k \bgroup et al\mbox.\egroup
  }{2017}]{Moravkeaam6960}
Morav{\v c}{\'\i}k, M.; Schmid, M.; Burch, N.; Lis{\'y}, V.; Morrill, D.; Bard,
  N.; Davis, T.; Waugh, K.; Johanson, M.; and Bowling, M.
\newblock 2017.
\newblock Deepstack: Expert-level artificial intelligence in heads-up no-limit
  poker.
\newblock {\em Science}.

\bibitem[\protect\citeauthoryear{Munoz~de Cote \bgroup et al\mbox.\egroup
  }{2013}]{MunozdeCote}
Munoz~de Cote, E.; Stranders, R.; Basilico, N.; Gatti, N.; and Jennings, N.
\newblock 2013.
\newblock Introducing alarms in adversarial patrolling games: Extended
  abstract.
\newblock In {\em AAMAS}.

\bibitem[\protect\citeauthoryear{Raghavan and
  Thompson}{1987}]{raghavan1987randomized}
Raghavan, P., and Thompson, C.~D.
\newblock 1987.
\newblock Randomized rounding: a technique for provably good algorithms and
  algorithmic proofs.
\newblock {\em Combinatorica} 7(4):365--374.

\bibitem[\protect\citeauthoryear{Tambe}{2011}]{Tambe2011}
Tambe, M.
\newblock 2011.
\newblock {\em Security and Game Theory: Algorithms, Deployed Systems, Lessons
  Learned}.

\bibitem[\protect\citeauthoryear{Tsai \bgroup et al\mbox.\egroup
  }{2010}]{tsai2010urban}
Tsai, J.; Yin, Z.; Kwak, J.-y.; Kempe, D.; Kiekintveld, C.; and Tambe, M.
\newblock 2010.
\newblock Urban security: Game-theoretic resource allocation in networked
  physical domains.
\newblock In {\em AAAI}.

\bibitem[\protect\citeauthoryear{USDOT}{2016}]{USDOT}
USDOT.
\newblock 2016.
\newblock Bureau of transportation statistics.
\newblock {\small
  \url{https://www.transtats.bts.gov/Data_Elements.aspx?Data=2}}.
\newblock Accessed: 2017-02-15.

\bibitem[\protect\citeauthoryear{Xu}{2016}]{Xu16a}
Xu, H.
\newblock 2016.
\newblock The mysteries of security games: Equilibrium computation becomes
  combinatorial algorithm design.
\newblock In {\em The 17 ACM Conference on Economics and Computation (ACM-EC)}.

\end{thebibliography}

\clearpage
\appendix
\section{Appendix}

\subsection{ARA complex model} It is easy to change the probability to be $c_t = f(\{x_{i,j} ~|~ (i,j) \in T\})$ so that the probability is non-linear (as given by f) in the entries. Also, clearly increasing number of targets to be large is allowed as targets can be any subset of the indexes. A general sum variant can also be stated by specifying the adversary utilities in terms of $c_t$. Of course, as in SSGs, the general sum case cannot be written as a LP but can be written as multiple optimizations or a MILP.

\subsection{Implementability}
Viewing SSGs as ARAs provides an easy way of determining implementability using results from randomized allocation~\cite{budish2013designing}. First, we define 
\emph{bi-hierarchical assignment constraints} as those that can be partitioned into two sets $H_1, H_2$ such that two constraints $S, S'$ in the same partition ($H_1$ or $H_2$) it is the case that either $S \subseteq S'$ or $S' \subseteq S$ or $S \cap S' = \phi$. Then, we obtain the following sufficiency result
\begin{proposition}
All marginal strategies are implementable, or more formally $conv(P) = MgS$, if the assignment constraints are bi-hierarchical. 
\end{proposition}
\begin{proof} The proof is just an application of Thm 1 in~\cite{budish2013designing}.
\end{proof}

In contrast to prior work that have identified cases of non-implementability~\cite{KorzhykCP10,letchford2013solving,TSG2016} for specific cases, this provides an easy way to characterize non-implementability across a wide range of SSGs.
As Figure~\ref{Illus} reveals, both FAMS and TSG have non-implementable marginals due to overlapping constraints. 
Next, as defined in~\cite{budish2013designing}, \emph{canonical assignment constraints} are those that impose constraints on all rows and columns of the matrix (with possibly additional constraints), using which we obtain the following necessity result
\begin{proposition}
Given canonical assignment constraints, if all marginal strategies are implementable then the assignment constraints are bi-hierarchical.
\end{proposition}
\begin{proof} The proof is just an application of 2 in~\cite{budish2013designing}.
\end{proof}

\subsection{Proofs}
\textbf{Proof of Theorem~\ref{ARAHard}}:

\begin{proof} For the first part, given any ARA problem with a NP hard DBR-U instance (for the decision version of DBR-U), we construct another ARA instance such that the feasibility problem for that ARA instance solves the DBR-U decision problem. Thus, as the feasibility is NP Hard, there exists no approximation. First, since ARA problem is so general there exists DBR-U problems that are NP Hard. For example the DBR-U problem for FAMS has been shown to be NP Hard~\cite{Xu16a}. Given the hard DBR-U problem, form an ARA problem  with by adding the constraint $\mathbf{1} \cdot \mathbf{x} = k$. Also, let there be only one target $t$ in the problem, so that the objective becomes $U(\mathbf{x}, t)$ insted of $z$ and all constraints in the optimization are just the marginal space constraints and $\mathbf{1} \cdot \mathbf{x} = k$. Now, the existence of any solution of the optimization gives a feasible point as $\mathbf{x} = \sum_m a_m \mathbf{x_m}$, where the integral points $x_m$ is solution to the question of does there exist a solution of the DBR-U optimization problem with value $k$. Next, a binary search on $k$ can answer the decision DBR-U problem of whether there exists a solution for decision DBR-U with value $\geq k$. Thus, existence of any solution for ARA is NP Hard, thus, no approximation is easy.

For the second part, we present a AP approximation preserving reduction (with problem mapping that doe not depend on approximation ratio); such a reduction preserves membership in PTAS, APX, log-APX, Poly-APX (see ). 
Given any DBR problem, we construct the ARA problem with one target such that $T = \{1, \ldots, k\}\times \{1, \ldots, n\}$. Choose the weights $w_{i,j}$'s such that $w_{i,j} \propto d_{i,j}$ and $w_{i,j} \leq 1/\max_{\mathbf{x} \in MgS} \sum_{i,j} x_{i,j}$. Observe that $\max_{\mathbf{x} \in MgS} \sum_{i,j} x_{i,j}$ is computable efficiently and $\max_{\mathbf{x} \in MgS} \sum_{i,j} x_{i,j} \geq \max_{\mathbf{x} \in conv(P)} \sum_{i,j} x_{i,j}$, thus, the ARA is well-defined. Thus, due to just one target, the ARA optimization is same as $\max_{\mathbf{x} \in conv(P)} \mathbf{w} \cdot \mathbf{x}$. Suppose we can solve this problem with $r$ approximation with the solution mixed strategy being $
\mathbf{x}^\epsilon = \sum_{i=1}^m a_i \mathbf{x_i}$ for some pure strategies $\mathbf{x_i}$. Now, since $w_{i,j} \propto d_{i,j}$ we also know that this solution also provides $r$ approximation for DBR-C. Let the optimal solution for DBR-C be $OPT$; note that $OPT$ is also the optimal solution for DBR. $
\mathbf{x}^\epsilon$ provides a solution value $\mathbf{w} \cdot \mathbf{x}^\epsilon \geq OPT/r$. Further, as the objective is linear in $\mathbf{x}$ and $
\mathbf{x}^\epsilon = \sum_{i=1}^m a_i \mathbf{P_i}$, it must be the case that there exists a $j \in \{1, \ldots, m\}$ such that $\mathbf{w} \cdot \mathbf{P_j} \geq \mathbf{w} \cdot \mathbf{x}^\epsilon \geq OPT/r$.  Thus, since $\mathbf{P_j} \in P$, $\mathbf{P_j}$ provides $r$ approximation for DBR. Since, $m$ the number of the pure strategies in support of $\mathbf{x}^\epsilon$ is polynomial, $\mathbf{P_j}$ can be found in polynomial time by a linear search.
\end{proof}

\textbf{Proof of Theorem~\ref{TSGInapp}}:
\begin{proof}
Given an independent set problem with $V$ vertices, we construct a TSG with  $\{1, \ldots, V + 1\}$ team types, where each team type in $1, \ldots, V$ corresponds to a vertex. The $V+1$ team is special in the sense that it does not correspond to any vertex and it is made up of just one resource with a very large resource capacity $2V$. Construct just one passenger category with passengers $N = V+1$. Since, there is just one passenger category (and target) we will use $x_i$ as the matrix entries instead of $x_{i,j}$. Choose $U^t_s =  V+1$ and $U^t_u = 0$ and
efficiencies $E_i = 1$ for all teams, except $E_{V+1} = 0$. Then, the objective of the integer LP is $\sum_{i=1}^V x_i = \mathbf{1}_V \cdot \mathbf{x}$ where $ \mathbf{1}_V $ is a vector with first $V$ components as 1 and last component as 0. 

Next, have resources for every edge $(i,k) \in E$ with resource capacity $1$. This provides the inequality $\sum_{(i,k) \in E} x_i + x_j \leq 1$. Also, we have $x_{V+1} \leq 2V$. Inspection of every passengers provides the constraints $\sum_{i=1}^{V+1} x_i = V+1$. Treating $x_{V+1}$ as a slack, we can see that the constraint $x_{V+1} \leq 2V$ and $\sum_{i=1}^{V+1} x_i = V+1$ are redundant. For the left over constraints $\sum_{(i,k) \in E} x_i + x_j \leq 1$, we can easily check that any valid integral assignment (pure strategy) is an independent set. Moreover, the objective $\sum_{i=1}^V x_i $ tries to maximize the independent set. The optimal value of this optimization over $conv(P)$ is an extreme point which is integral and equal to the maximum independent set OPT.  Thus, suppose a solution $\mathbf{x}^\epsilon$ to the SSE problem with value $\geq OPT/r$. Further, as the objective is linear in $\mathbf{x}$ and $
\mathbf{x}^\epsilon = \sum_{i=1}^m a_i \mathbf{P_i}$, it must be the case that there exists a $j \in \{1, \ldots, m\}$ such that $ \mathbf{1}_V \cdot\mathbf{P_j} \geq  \mathbf{1}_V \cdot \mathbf{x}^\epsilon \geq OPT/r$.  Thus, since $\mathbf{P_j} \in P$, $\mathbf{P_j}$ provides $r$ approximation for maximum independent set. Since, $m$ the number of the pure strategies in support of $\mathbf{x}^\epsilon$ is polynomial, $\mathbf{P_j}$ can be found in polynomial time by a linear search.

\end{proof}

\textbf{Proof of Theorem~\ref{FAMSHard}}:

\begin{proof}

We provide an AP reduction from independent set. As max independent set is poly APX complete, this rules out any log factor approximation.

Given an independent set maximization problem with vertices $V$ and edges $E$ construct the following FAMS problems, one for each $k$. Use $2n - k$ resources. All resources can be assigned to any schedule. Construct schedules $s_1, \ldots, s_n$ corresponding to the vertices $v_1, \ldots, v_n$. Construct target $t_{e}$ corresponding to every edge $e = (u,v)$ such that $t_e \in s_u$ and $t_e \in s_v$. All $t_e$'s have the same value for being defended or undefended and that value is $n+2$; thus, these targets do not need to covered but impose the constraint that $s_u$ and $s_v$ cannot be simultaneously defended. Thus, it is clear that any allocation of resources to $s_1, \ldots, s_n$ corresponds to an independent set. Next, consider additional $2n$ targets and expand the schedules such that $t_{i}, t_{i+1} \in s_{i}$. Further, add more singleton schedules $s_{n+1}, \ldots, s_{3n}$ with $t_{i} \in s_{n+i}$. All additional targets $t_1, \ldots, t_{2n}$ provide value $k$ when defended and $k - 2n$ otherwise. Thus, the expected utility of defending an additional target $t$ given coverage $c_t$ is $c_t(k) + (1-c_t)(k-2n) = 2n*c_t + k - 2n$. 

First, assume we have a poly time algorithm to approximately compute the SSE with approx factor $r$ ($r > 1$). We will run this poly time algorithm with resources $2$ to $2n-1$ which is again a poly time overall, and also the overall output size is poly. For the given max independent set problem, let the solution be $k^*$. Observe that for problems with resources $2n - k$ where $k \leq k^*$, all valuable (additional) targets can be covered by covering $k^*$ schedules with $2k^*$ targets in $s_1, \ldots, s_n$  and using the remaining $\geq 2n - 2k^*$ resources to cover the remaining $2n - 2k^*$ targets. This provides utility of $k$ for the SSE. In particular, the utility with $2n -k^*$ resources is $k^*$. Also note that for every problem, there is always a trivial allocation of $2n-k$ resources to the $2n$ singleton schedules such that coverage of each target is $1 - k/2n$. (this is deducible as the allocation to singleton schedules is unconstrained and can be implemented in poly time by Birkhoff-von Neumann result as provided in~\cite{KorzhykCP10}). This trivial allocation provides an utility 0.

Also the following result will be useful: given approximation factor $r$ for the case with $2n - k^*$ resources then one of the pure strategy output for this case will have at least $k^* - l_{\min}$ schedules among $s_1, \ldots, s_n$ covered where $l_{\min} = \lfloor \arg\!\min_l \frac{k^*}{k*-l} > r \rfloor$. Before proving the above note that by definition $\frac{k^*}{k*-l_{\min} - 1} > r$ and $\frac{k^*}{k*-l_{\min}} < r$. 

To prove the result in last paragraph consider the contra-positive: suppose all pure strategies output cover at most $k^* - l_{\min}-1 $ schedules among $s_1, \ldots, s_n$, then in each pure strategy at least $l_{\min}+1$ targets are not covered (since 2 targets are covered for the $k^* - l_{\min}-1 $ schedules and rest of resources can cover only 1 target). Then the coverage of the least covered target in the mixed strategy formed using such pure strategies is $\leq 1 - (l_{\min}+1)/2n$. The utility for this least covered target is $\leq k^* - l_{\min} - 1$. The overall utility has to be lower than utility for any target, hence the utility is $\leq k^* - l_{\min} - 1$. The optimal utility is $k^*$. Thus, by definition of approximation ratio $r$ we must have $k^* - l_{\min} - 1 \geq k^*/r$ or re-arranging $\frac{k^*}{k^*-l_{\min} - 1} \leq r $ but by definition of $l_{\min}$ we must have $\frac{k^*}{k*-l_{\min} - 1} > r$ hence a contradiction.

The above result also shows that at least one pure strategy yields an independent set of size $k^* - l_{\min}$. As $k^*$ is the max size of independent sets we obtain and approximation ratio $r'$ for the max independent set problem such that $r' < \frac{k^*}{k^* - l_{\min}} < r$. Thus, we obtain a better approximation $r'$ given $r$ approximation for the SSE. Thus, we have an AP reduction.
\end{proof}


\textbf{Proof of Theorem~\ref{FAMSApp}}:
\begin{proof}
Consider the event of a target $t$ having an infeasible assignment after the comb sampling. Call this event $E_t$. Let $C_{t,i}$ be the event that resource $i$ covers this target $t$. Then, $P(E_t) = \sum_{i} P(E_t|C_{t,i})P(C_{t,i})$.  From the guarantees of comb sampling we know that $P(C_{t,i}) = \sum_{j: (i,j) \in T} x^m_{i,j} \leq 1$ and $P(x_{i,j} = 1) = x^m_{i,j}$. Also, by comb sampling if $x_{i,j} = 1$ then $x_{i,j'} = 0$ for any $j'\neq j$. Next, we know that $P(E_t|C_{t,i})$ is the probability that the any of the other $x_{i',j}$ is assigned a one, which is $1 - $ the probability that all other $x_{i',j}$ are assigned 0. Thus, 
$$P(E_t|C_{t,i}) = 1 - \prod_{i' \neq i} (1- P(C_{t,i})) $$
Let $p_{t,i} = P(C_{t,i})$. Considering the fact that $\prod_i (1 - p_{t,i}) > 1 - \sum_i p_{t,i}$, we get 
$$1 - \prod_{i' \neq i} (1- P(C_{t,i})) \leq \sum_{(i',j): i'\neq i \land (i',j) \in T} x^m_{i',j} \leq 1 - \sum_j x^m_{i,j}$$
where the last inequality is due to the fact that $\sum_{(i,j) \in T} x^m_{i,j} \leq 1$. 

Thus, $P(E_t) \leq \sum_{i} (1 - p_{t,i})p_{t,i} \leq \sum_{i} p_{t,i} - \sum_{i} (p_{t,i})^2$. 
Next, we know from standard sum of squares inequality that $\sum_{i} (p_i)^2 \geq (\sum_{i} p_i)^2/k$. Thus, we get $P(E_t) \leq (\sum_{i} p_i) (1 - \sum_{i} p_i/k)$ The RHS is maximized when $\sum_{i} p_i = 1$, thus, $P(E_t) \leq 1 - 1/k$. Also, then $P(\lnot E_t) \geq 1/k$

Now consider the coverage of target $t$: $x^m_t = \sum_{(i,j) \in T} x^m_{i,j}$. According to our algorithm the allocation for target $t$ continues to remain $1$ with probability $(1/2)^C$ if its allocation is already feasible after comb sampling (and we always obtain a pure strategy). This is because this target shares schedules with $C$ other targets and thus in the worst case may be reduced with $1/2$ probability for each of the $C$ targets. We do a worst case analysis and assume that no resource is allocated to a target when the sampled allocation is infeasible for that target. Thus, let $y_t$ denote the random variable denoting that target $t$ is covered. Thus, $E(y_t) = P(y_t = 1) = P(y_t = 1|E_t) P(E_t) + P(y_t = 1|\lnot E_t)P(\lnot E_t)$. Now, $P(y_t = 1|\lnot E_t)$ is same as $x^m_t/2^C$ and we assumed the worst case of $P(y_t = 1|E_t) = 0$. Thus, we have $E(y_t) \geq x^m_t/2^Ck$. As the utilities are linear in $y_t$, we have the utility for $t$ as $U_t \geq U_t^m/2^Ck$, where $U_t^m$ is the utility under the marginal $\mathbf{x}^m$. Thus, if $t^*$ is the choice of adversary under the marginal $\mathbf{x}^m$ we know that $U_{t^*}^m$ is the lowest utility for the defender over all targets $t$. Hence, we can conclude that the utility with the approximation is at least $U_{t^*}^m/2^Ck$
\end{proof}

\textbf{Proof of Theorem~\ref{TSGApp}}:

\begin{proof}
The main assumption in the proof is that the steps after after comb sampling changes the probability of detecting an adversary in passenger category $j$ by at most $1/c$. Also, by assumption of the theorem since Algorithm~\ref{abstractalgorithm} does not fail ever, the change in  utility for any passenger category $j$ is at most a factor of $1/c$. By similar reasoning as for FAMS, we conclude that this provides a $c$-approximation.
\end{proof}


\end{document}